\documentclass{llncs}

\usepackage{graphicx}

\title{Enumeration of Sequences with Large Alphabets}
\titlerunning{Enumeration of $\sigma$--ary sequences}

\author{M. O\u guzhan K\"ulekci}
\institute{T\"UB\.ITAK - B\.ILGEM \\
National Research Institute of Electronics and Cryptology\\
\email{oguzhan.kulekci@tubitak.gov.tr}}

\authorrunning{M.O. K\"ulekci}

\begin{document}
\maketitle

\begin{abstract} 
A binary sequence of length $n$ with $w$ ones can be identified by its lexicographical rank in the set of all binary
sequences with same number of ones and zeros, which is of size $\frac{n!}{w!\cdot (n-w)!}$.  Although that
enumeration has been deeply studied for binary case, it is less addressed for $\sigma$-ary sequences, where $\sigma>2$.
 
Assuming $n$ is a fixed predetermined parameter, the enumerative coding of a given $n$-symbols long sequence $S$ over
alphabet $\Sigma=\{\epsilon_1, \epsilon_2, \ldots, \epsilon_\sigma\}$ consists of first specifying the frequencies of
characters in $S$ denoted by $C=\langle c_1, c_2, \ldots, c_\sigma \rangle$, followed by indicating the rank of $S$
among all possible sequences with the same $C$ vector.  
The required code lengths are  $(\sigma-1)\lceil \log (n+1) \rceil$ bits for $C$ and 
$\lceil \log \frac{n!}{c_1! \cdot c_2! \cdot \ldots \cdot c_\sigma!} \rceil$ bits for the rank. 

This study focuses on efficient schemes for enumerative coding of $\sigma$--ary sequences by mainly borrowing ideas
from \"Oktem \& Astola's \cite{Oktem99} hierarchical enumerative coding  and Schalkwijk's \cite{Schalkwijk72}
asymptotically optimal combinatorial code on binary sequences. 

By observing that the number of distinct $\sigma$--dimensional vectors having an inner sum of $n$, where the values in
each dimension are in range $[0\ldots n]$  is  $K(\sigma,n) = \sum_{i=0}^{\sigma-1} {{n-1} \choose {\sigma-1-i}}
{{\sigma} \choose {i}}$,  we propose representing $C$ vector via enumeration, and present necessary
algorithms to perform this task. 
We prove $\log K(\sigma,n)$ requires approximately $ (\sigma -1) \log (\sigma-1) $ less bits than
the naive 
$(\sigma-1)\lceil \log (n+1) \rceil$ representation for relatively large $n$, and examine the results for varying
alphabet sizes
experimentally. 

We extend the basic scheme for the enumerative coding of $\sigma$--ary sequences by introducing a new method
for large alphabets. We experimentally show that the newly introduced technique is superior
to the basic scheme by providing experiments on DNA sequences. 

\end{abstract}

\section{Introduction}

Let $S$ be a finite sequence of length $n$, where symbols are drawn from alphabet $\Sigma=\{\epsilon_1,\epsilon_2,\ldots
, \epsilon_\sigma\}$ and the number of occurrences  of individual characters are indicated by the vector  $C=\langle
c_1, c_2, \ldots ,c_\sigma\rangle$.

The $k$th-order \emph{empirical entropy} of $S$ has been defined by Manzini \cite{Manzini2001} as
$$H_k(S) = \frac{1}{n}\sum_{\forall w \in \Sigma^k} |w_S| \cdot H_0(w_S)$$
, where $w_S$ represents the string formed
by concatenation of all symbols following the $k$--symbols long context $w$ in $S$, and $H_0$ is the zeroth-order
entropy defined by $H_0(s)= - \sum_{ \forall \epsilon \in \Sigma} P(\epsilon) \cdot \log P(\epsilon)$. Note that
$P(\epsilon)$ is the probability of symbol $\epsilon$ in $s$.

On the other hand, Grossi \emph{et.al} \cite{finiteset} proposed \emph{finite set entropy} that approximates the $k$-th
order information content of $S$ as 
$$H_k(S) = \frac{1}{|S|} \sum_{\forall w \in \Sigma^k} \log \bigl( \frac{c_w!}{c_{w\epsilon_1!} \cdot c_{w\epsilon_2!}
\ldots c_{w\epsilon_\sigma!} }\bigr) $$
, where $c_w$ is the frequency of $k$--symbols long context and similarly $c_{w\epsilon_i}$ is the observation count of
$\epsilon_i$ following context $w$ in $S$.  The zeroth order finite set entropy of $S$ is then 
$H_0(S) = log \bigl( \frac{n!}{c_1! \cdot c_2! \cdot \ldots \cdot c_\sigma! } \bigr)$. 

The idea behind finite set entropy is that $S$ can be identified by its rank in the sorted list of all sequences having
the same frequencies of individual symbols. 
Similarly, given the position index and the frequencies, it is possible to generate the original sequence $S$. 

The enumerative coding is more related with the finite set entropy, where Huffman or arithmetic coding is computed
according to the empirical entropy estimation via a Markovian approach in practice.

The number of ones, $w$, in a given binary sequence of length $n$ can be specified by using $\lceil \log (n+1) \rceil$ 
bits  as $w$ can take one of $(n+1)$ values in range $[0\ldots n]$. 
The list of all unique length--$n$ binary sequences that have exactly $w$ ones and $n-w$ zeros is of size 
$\big(\frac{n!}{(n-w)!\cdot w!}\big)$. 
Any sequence among this set can be identified by indicating its lexicographical rank, which requires 
$\lceil \log \big(\frac{n!}{(n-w)!\cdot w!}\big) \rceil$ bits. 

Lynch \cite{Lynch66} and  Davisson \cite{Davisson66} proposed that scheme to represent binary sequences, which is 
named as enumerative or combinatorial coding. 
More general frameworks has been investigated by Cover\cite{Cover73} and Rissanen\cite{Rissanen83}. 
A comparison against  arithmetic coding has been given by Cleary \& Witten \cite{Cleary84}. 

Although $\log \big(\frac{n!}{(n-w)!\cdot w!}\big)$ is optimum due to the entropy of a finite set, the 
$\lceil \log (n+1) \rceil$ bits needed to code the weight information cause an overhead. 
This overhead becomes more significant in case of $\sigma$--ary sequences. 
The naive encoding of the frequency vector $C$ requires  $(\sigma-1)\lceil \log (n+1) \rceil$ bits assuming $n$ is fixed
predetermined parameter, which otherwise necessitates  $\sigma \lceil \log (n+1) \rceil$ bits in case of variable $n$. 
Several studies have addressed this fact and proposed alternatives for binary case. However, enumerative coding of
sequences over large alphabets has not received much attention to date. 

\textbf{Our Contribution.} This study focuses on enumeration of $\sigma$--ary sequences.
We show that the number of distinct $\sigma$--dimensional vectors having a constant inner sum of $n$, where the values
in each dimension are in range $[0\ldots n]$,  is $K(\sigma,n) = \sum_{i=0}^{\sigma-1} {{n-1} \choose {\sigma-1-i}}
{{\sigma} \choose {i}}$. We have noticed that \"Oktem\&Astola \cite{Oktem99} previously indicated this
number\footnote{However, we could not find the proof or a derivation of the formula in their regarding paper.} as
${n+\sigma-1} \choose {\sigma -1}$, which is in concordance with our finding according to the
Chu-Vandermonde  identity \cite{Koepf98}. 
We prove $\log K(\sigma,n)$ requires approximately $ (\sigma -1) \log (\sigma-1)$ less bits than  
$(\sigma-1)\lceil \log (n+1) \rceil$ for relatively large $n$, and examine the results for varying alphabet sizes
experimentally. 

We extend the basic scheme for the enumerative coding of $\sigma$--ary sequences by introducing a new method, that
borrows ideas from Schalkwijk's \cite{Schalkwijk72} enumerative coding of binary sequences via variable length blocks. 
We experimentally show that the newly introduced technique is superior to the basic scheme by providing experiments on
DNA sequences. 

The results indicated herein are of significance not only in terms of enumerative coding, but also for the related
fields such as enumerative combinatorics \cite{Stanley}, text indexing\cite{finiteset,Martin10}, information
theoretic analyses of sequences \cite{infotheory}, and jumbled pattern matching \cite{jsm}. 

The outline of the paper is as follows. We investigate the related previous studies in section $2$. We introduce the
compact representation of frequency vectors in section $3$ with the necessary enumeration algorithms. In section $4$ we
describe the proposed $q$--ary enumeration scheme based on factorizing the input sequence into variable--length blocks.
We give the implementation details and experimental results on DNA sequences in section $5$ and conclude by section $6$.

\section{Related Work}

Since the studies concerning enumerative coding had generally focused on binary sequences, previous efforts for the
efficient representation of weight $w$ had also appeared mainly on the binary alphabet.

In the scheme proposed by Schalkwijk  \cite{Schalkwijk72} the sequence is partitioned into variable length blocks,
where each block contains either $w$ ones or $(n-w)$ zeros. During the partitioning, the sequence is
scanned from left to right and whenever either the specified number of zeros or ones is obtained, the process stops.
The subsequence is then generated by concatenating the symbols since the last cut point  and   hypothetically padded
with the remaining ones or zeros, so that its length becomes $n$. 
The rank in the lexicographically sorted list of all binary $n$-bits long sequences with $w$ ones uniquely identifies 
that subsequence. The partitioning and coding continues in the same way until the end of the input is reached.

As an example, assume
the sequence is $1100100010100111$, and the parameters are $w=3$, $n=5$. 
The subsequences with padded bits shown between brackets are 
$1100$\{$1$\}, $100$\{$11$\}, $010$\{$11$\}, $100$\{$11$\}, and $111$\{$00$\}, which can be denoted by their
lexicographic ordering among the $\frac{5!}{3!\cdot 2!} = 10$ possible cases of $00\{111\}$, $010\{11\}$, $0110\{1\}$,
$0111\{0\}$, $100\{11\}$, $1010\{1\}$, $1011\{0\}$, $1100\{1\}$, $1101\{0\}$, $111\{00\}$. In the decoding phase,
fixed-length code words are read, and mapped to their corresponding sequences. It is easy to detect the padded bits
again with the same methodology of factorization as once $w$ ones or $(n-w)$ zeros are observed in a subsequence, it can
be determined that the rest of the string would be composed of all ones or zeros that are not part of the original
sequence.

The length of the code words are fixed and the blocks are variable in Schalkwijk's proposal. 
Hence, $w$ and $n$ are predetermined parameters that does not need to be coded separately in each block, which brings
an advantage.  
On the other hand, we have an overhead due to the padded bits. 
However, the enumeration scheme of Schalkwijk is 
shown to be asymptotically optimum for binary memoryless channels \cite{Schalkwijk72}. 
Another difficulty in practice is finding the optimum values of $w$ and $n$.

\"Oktem \& Astola \cite{Oktem99} introduced  hierarchical enumerative coding for binary sequences. They divide the
sequence into smaller subsequences. The main idea is that once the total number of ones in the larger sequence is
known, the weights of the smaller length blocks can be coded more efficiently. 
For example, we divide the sample binary sequence $1100100010100111$ into four equal length blocks as $1100$, $1000$,
$1010$, $0111$, where the weights of
each block are $2$, $1$, $2$, and $3$ respectively. If we know in advance that the main sequence is of
length $16$ with a total weight of $8$ and each substring is $4$ bits long, then we can count the ways that $4$
non-negative integers sum up to $8$, and simply represent the vector $\langle 2,1,2,3 \rangle$ via its rank in that
set. They proposed to perform this in a tree structure, where the weights of the child nodes are determined by their
ancestors (see \cite{Oktem99,Oktemthesis} for detailed discussions), and hence, given name hierarchical enumeration.  

Another study concerning to reduce the overhead of transmitting the weight parameter has been presented by Dai\&Zakhor
\cite{Dai03}. In their work the authors proposed to model the distribution of $w$  values via a Poisson process, and
then code individual weights by Huffman coding according to their probabilities in the estimated distribution. 
 
\section{Enumerating the frequency vectors}

The overhead in enumerative coding of $S$ is the information required to represent the individual symbol frequencies.
Assuming that the length $n$ of $S$ is predetermined, the naive way of representing the frequency vector 
$C = \langle c_1, c_2, \ldots, c_\sigma \rangle$ requires $(\sigma-1)\log (n+1)$ bits. Notice that $\sigma-1$ items are
enough as we can compute the last frequency by subtracting the total frequency of the remaining items from $n$.  

In this study we propose to represent $C$ vector via enumeration. 
Theorem $1$ with lemma $1$ below counts the distinct $\sigma$--dimensional vectors that the sum of the individual
components is predefined. In other words, we count the ways that $\sigma$ counters sum up to a specified value when
we know in advance each counter is greater than or equal to zero. 

\begin{lemma}
 The number of ways that $K$ counters, whose values are positive integers,  sum up to $N$ is 
\begin{equation}
 KsumN(K,N) = {{N-1} \choose {K-1}}.
\end{equation}
\end{lemma}
\begin{proof}
Assume we have a chain of $N$ balls connected by single wires. Obviously we have
$N-1$ connections. When we cut $K-1$ of them, we will be left with $K$ partitions. The number of balls on
each partition identifies the value of the corresponding counter. Since this selection can be performed
in $N-1 \choose K - 1$ ways, this is the total number of distinct  ways that $K$ counters sum up to $S$.
\end{proof}

\begin{theorem}
Let $\mathcal{V}= \langle v_1,v_2,\ldots,v_\sigma \rangle$ be a $\sigma$-dimensional vector of integers,  where $v_i
\geq 0$ for all $0<i\leq \sigma$ and $\sum_{i} v_i = n$. The number of such distinct vectors is 
\begin{equation} 
K(\sigma,n) = \sum_{i=0}^{\sigma-1} \left( \begin{array}{c} n-1 \\ \sigma-i-1 \end{array}\right) \left(
\begin{array}{c} \sigma \\ i \end{array}\right) = {{n+\sigma-1} \choose {\sigma-1}} 
\end{equation}
\end{theorem}
\begin{proof}
Among the $\sigma$ dimensions of $\mathcal{V}$, at least one of them is greater than zero and at most all of them are
non-zero. 
If we represent the number of zero values by $i$, the summation iterates over $i$ from $0$ to $\sigma - 1$. 
The second binomial coefficient $\sigma \choose i$ in the summation considers number of ways that we can choose $i$
from $\sigma$ dimensions, where the first binomial coefficient $n-1 \choose \sigma - i - 1$ counts the
distinct ways that $n$ items can be partitioned into  $\sigma-i$ urns. By the Chu-Vandermonde  identity \cite{Koepf98}
this summation becomes ${{n+\sigma-1} \choose {\sigma-1}}$.
\end{proof}

\begin{lemma}
Let $\mathcal{V}= \langle v_1,v_2,\ldots,v_\sigma \rangle$ be a $\sigma$-dimensional vector of integers,  where $v_i
\geq 0$ for all $0<i\leq \sigma$ and $\sum_{i} v_i = n$. The enumerative coding of $C$ requires approximately
$(\sigma-1) \log (\sigma-1)$ bits less then naive representation of $(\sigma -1) \log (n+1)$ for $\sigma>2$ and 
$n$ is relatively larger than $\sigma$.
\label{lemma:2}
\end{lemma}
\begin{proof}
 According to Stirling's approximation of $z! \approx \sqrt{2 \pi z} (\frac{z}{e})^z$, the $K(\sigma,n) = {{n+\sigma-1}
\choose {\sigma-1}}$ can be estimated as 
\begin{equation}
K(\sigma,n) \approx \sqrt{\frac{(n+\sigma-1)}{2 \pi n (\sigma-1)}} \cdot
\frac{(n+\sigma-1)^{(n+\sigma-1)}}{n^n \cdot (\sigma-1)^{(\sigma-1)}} 
\label{eq1}
\end{equation}

Rearranging the terms in equation \ref{eq1} yields 
\begin{equation}
K(\sigma,n) \approx \sqrt{ \frac{n}{2 \pi n (\sigma-1)} + \frac{(\sigma-1)}{2 \pi n (\sigma-1)}} 
\cdot
\bigl(\frac{n+\sigma-1}{n}\bigr)^n 
\cdot  
\bigl(\frac{n+\sigma-1}{\sigma-1}\bigr)^{\sigma-1}.  
\label{eq2}
\end{equation}

The square-root term is less than one, and hence, the logarithm of $K(\sigma,n)$ can be approximated as 
\begin{eqnarray}
\log K(\sigma,n) & \approx & n[\log (n+\sigma-1) - \log n]+(\sigma-1)[\log (n+\sigma-1)-\log (\sigma-1)] \nonumber \\
                 &  = & (n+\sigma-1)\log (n+\sigma-1) - n\log n - (\sigma-1)\log(\sigma-1)
\label{eq3}
\end{eqnarray}

Assuming $\log (n+\sigma-1) \approx \log n$, when $\sigma>2$ and $n$ is relatively larger than $\sigma$, equation
\ref{eq3} becomes 
\begin{eqnarray}
\log K(\sigma,n) & \approx & (n+\sigma-1)\log n - n\log n - (\sigma-1)\log(\sigma-1) \nonumber \\
                 & =       & (\sigma-1)\log n - (\sigma-1)\log(\sigma-1)
\label{eq4}
\end{eqnarray}

Thus, considering $\log(n+1) \approx \log n$, the enumerated representation of a frequency vector is 
approximately $(\sigma-1)\log (\sigma-1)$ bits less than the naive $(\sigma-1)\log(n+1)$ representation.
\end{proof}

Figure \ref{fig:1} compares number of bits required in enumerative coding of vectors versus naive
representation for alphabet sizes of $4$, $20$, $128$, and $256$. For varying sequence lengths of $n$, where $n<1000$,
$(\sigma-1)\log n$ versus $\log K(\sigma.n)$ values are plotted. The advantage obtained by the
proposed enumerative coding becomes more significant with the increasing alphabet sizes, and the gain tends to converge
to the theoretical value $(\sigma-1)\log (\sigma-1)$ as stated by Lemma \ref{lemma:2}.
\begin{figure}
\begin{center}
\begin{tabular}{cc}
 \includegraphics[scale = 0.40]{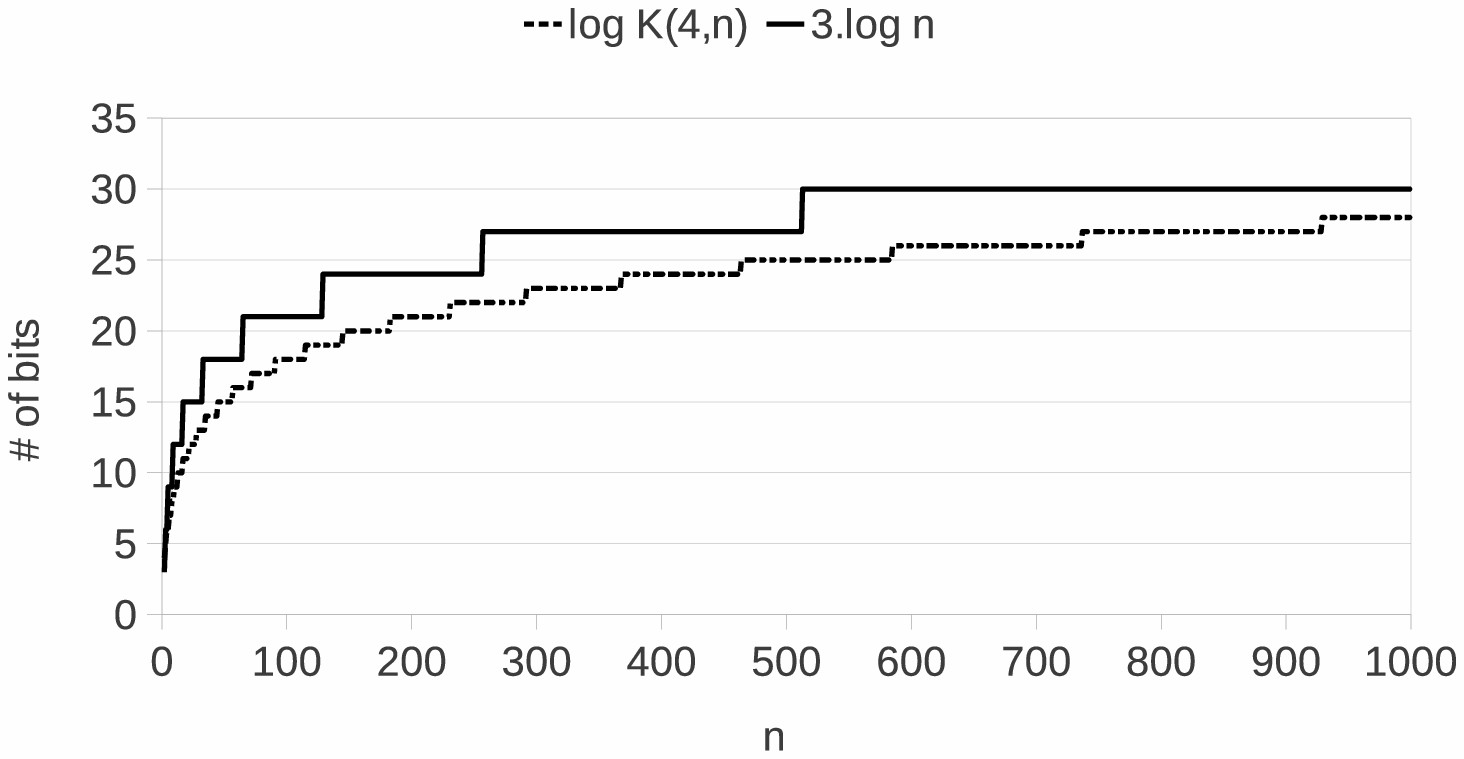} &
 \includegraphics[scale = 0.40]{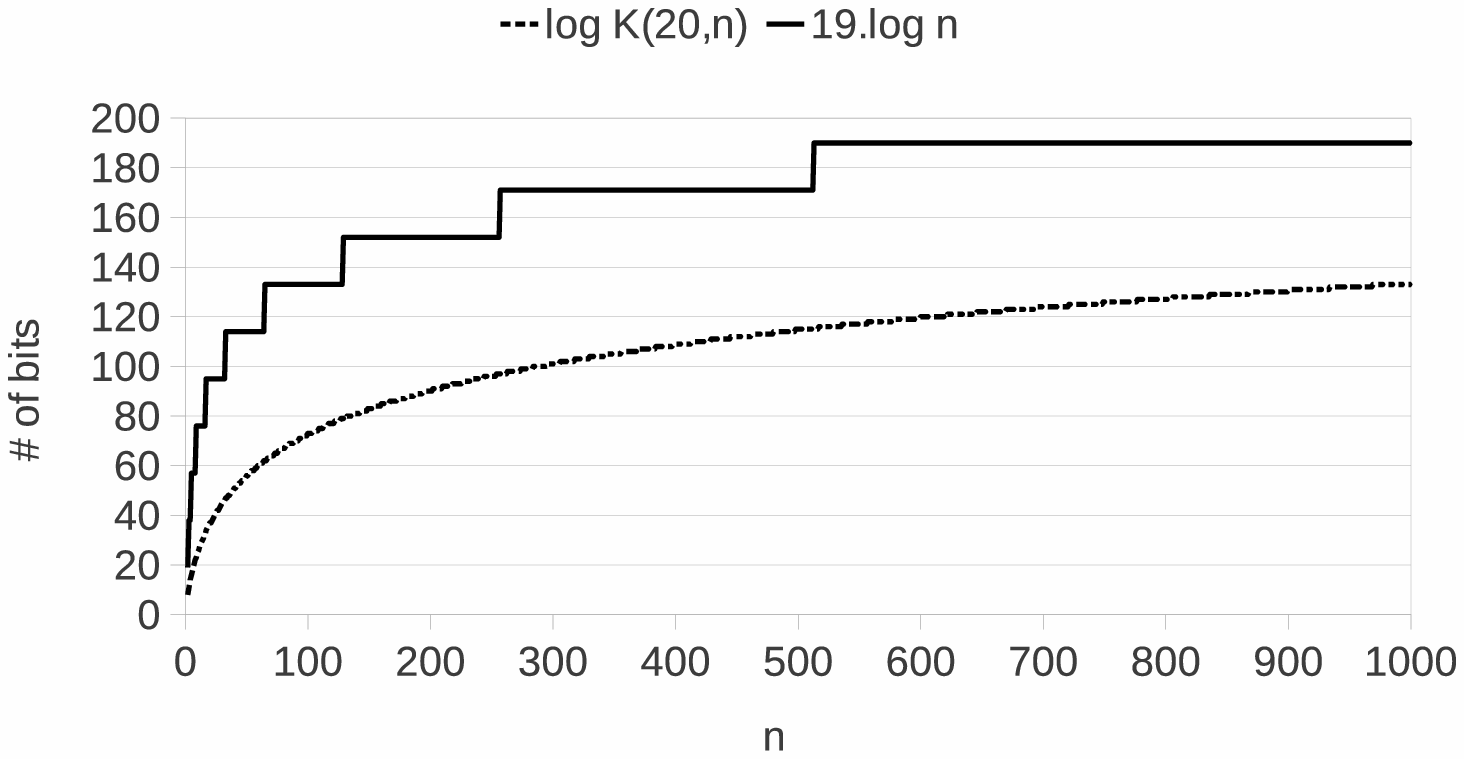} \\
a) $\sigma=4$ & b) $\sigma=20$ \\
 \includegraphics[scale = 0.40]{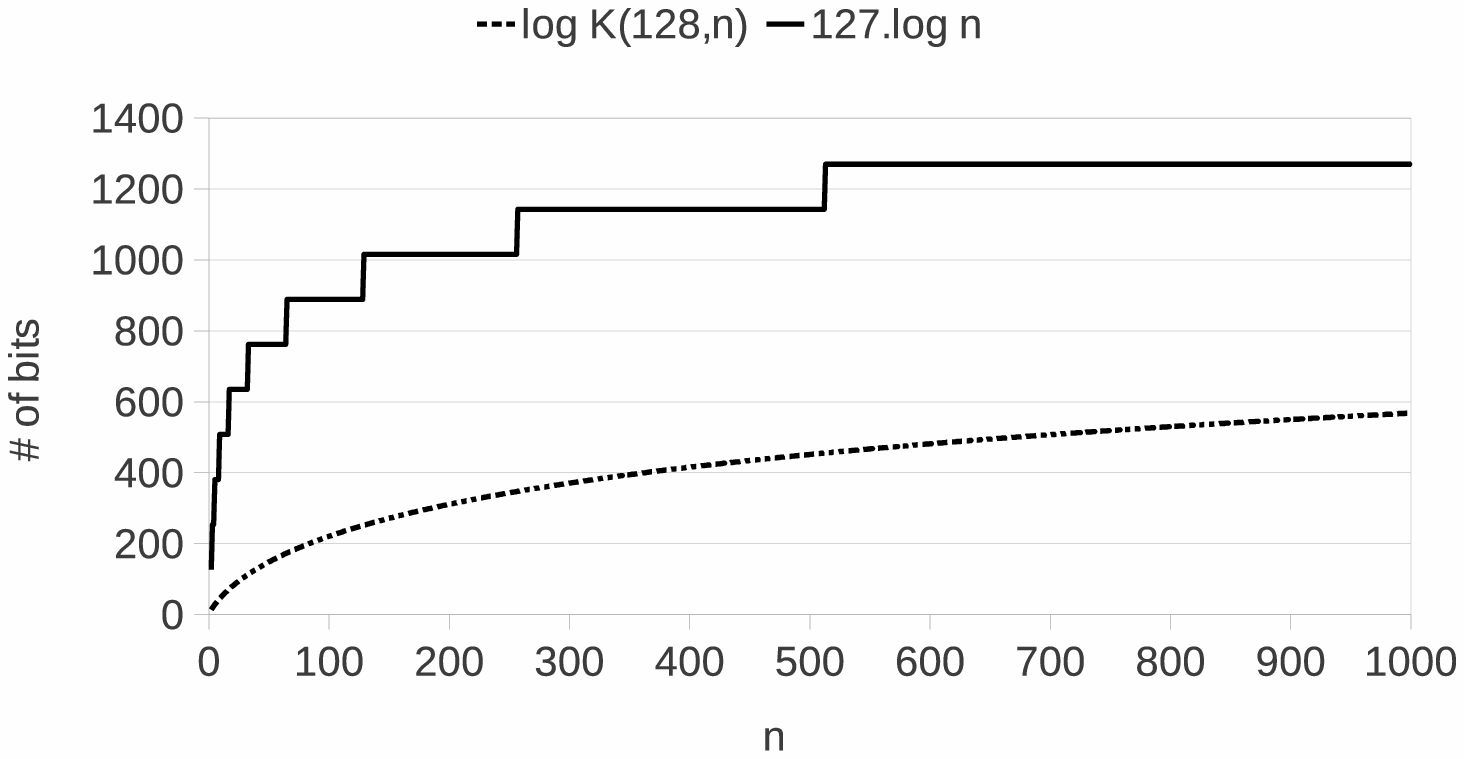} &
 \includegraphics[scale = 0.40]{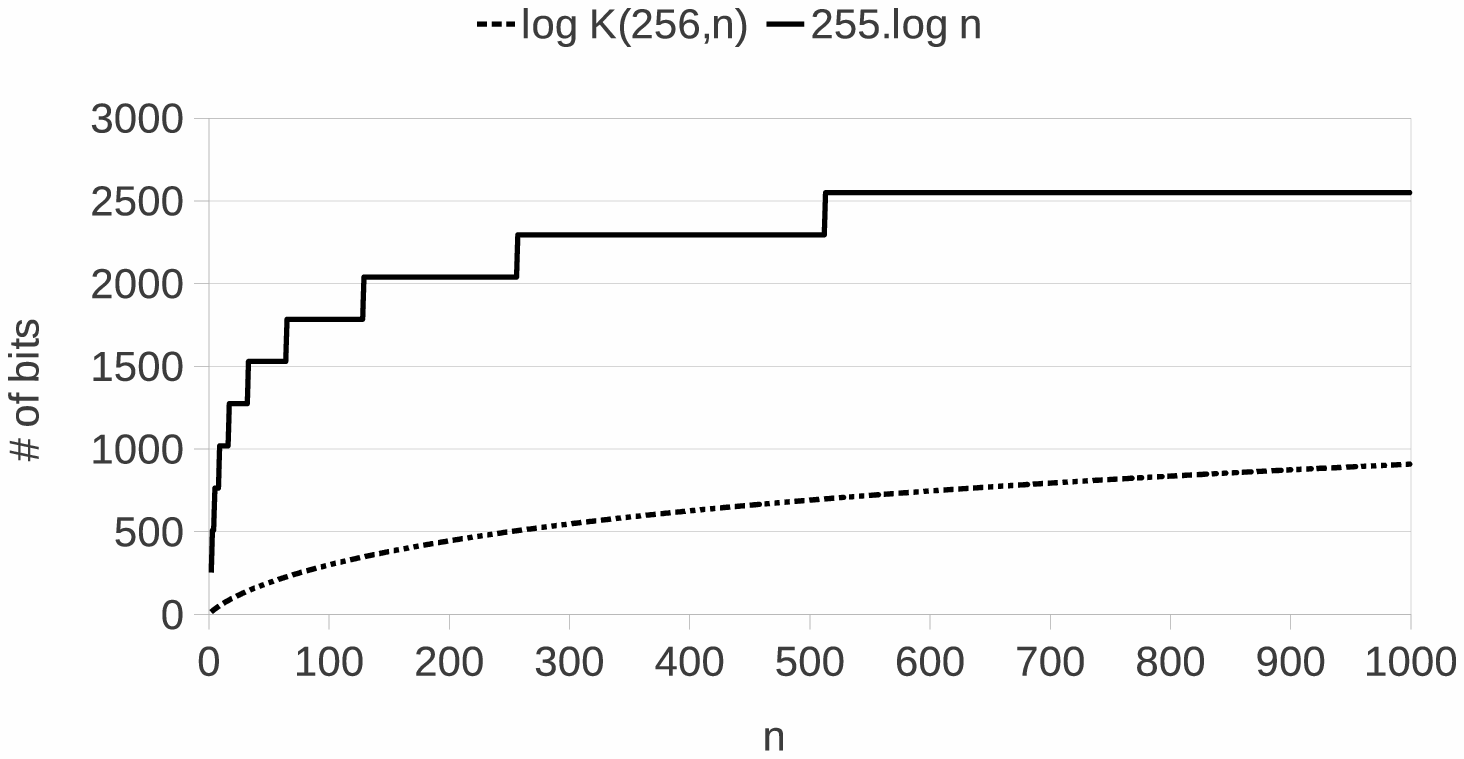} \\
a) $\sigma=128$ & d) $\sigma=256$ \\
\end{tabular}
\end{center}
\caption{Comparison of enumerated representation of frequency vectors versus naive representation in varying alphabet
sizes.}
\label{fig:1}
\end{figure}

The enumerative coding of the $C=\langle c_1, c_2, \ldots, c_\sigma \rangle$ requires a ranking among the
$\sigma$--dimensional vectors, whose inner sum is a predetermined value $n$. Once we have a mechanism to sort all such
vectors, we can identify any of them by indicating its rank among all possibilities.

We devise such a ranking, where $c_1$ and $c_{\sigma-1}$ are the most and least significant positions respectively.
The $c_\sigma$ value do not have an effect on the ordering since it is not a free variable as its value is determined
by $n - \sum_{1 \leq i< \sigma} c_i$. 

The individual vectors are enumerated in an ascending order according to their values of $c_i$'s. Table \ref{tab:1}
shows this ranking on an example case that we consider $4$ dimensional vectors having an inner sum of $4$. Notice that
there exist $K(4,4)=35$ of such vectors. 
\begin{table}
\renewcommand{\arraystretch}{1.2}
\begin{tabular*}{0.99\textwidth}{c@{\hspace{10pt}} @{\hspace{10pt}}c@{\hspace{10pt}}@{\hspace{10pt}}c@{\hspace{10pt}}
@{\hspace{10pt}}c@{\hspace{10pt}}@{\hspace{10pt}}c}

\begin{tabular}{c|c}
0 & $\langle 0,0,0,4 \rangle$\\ \hline
1 & $\langle 0,0,1,3 \rangle$\\ \hline
2 & $\langle 0,0,2,2 \rangle$\\ \hline
3 & $\langle 0,0,3,1 \rangle$\\ \hline
4 & $\langle 0,0,4,0 \rangle$\\ \hline
5 & $\langle 0,1,0,3 \rangle$\\ \hline
6 & $\langle 0,1,1,2 \rangle$\\ 
\end{tabular} 
&
\begin{tabular}{c|c}
7 & $\langle 0,1,2,1 \rangle$\\ \hline
8 & $\langle 0,1,3,0 \rangle$\\ \hline
9 & $\langle 0,2,0,2 \rangle$\\ \hline
10& $\langle 0,2,1,1 \rangle$\\ \hline
11& $\langle 0,2,2,0 \rangle$\\ \hline
12& $\langle 0,3,0,1 \rangle$\\ \hline
13& $\langle 0,3,1,0 \rangle$\\ 
\end{tabular} 
&
\begin{tabular}{c|c}
14& $\langle 0,4,0,0 \rangle$\\ \hline
15& $\langle 1,0,0,3 \rangle$\\ \hline
16& $\langle 1,0,1,2 \rangle$\\ \hline
17& $\langle 1,0,2,1 \rangle$\\ \hline
18& $\langle 1,0,3,0 \rangle$\\ \hline
19& $\langle 1,1,0,2 \rangle$\\ \hline
20& $\langle 1,1,1,1 \rangle$\\ 
\end{tabular} 
&
\begin{tabular}{c|c}
21& $\langle 1,1,2,0 \rangle$\\ \hline
22& $\langle 1,2,0,1 \rangle$\\ \hline
23& $\langle 1,2,1,0 \rangle$\\ \hline
24& $\langle 1,3,0,0 \rangle$\\ \hline
25& $\langle 2,0,0,2 \rangle$\\ \hline
26& $\langle 2,0,1,1 \rangle$\\ \hline
27& $\langle 2,0,2,0 \rangle$\\ 
\end{tabular}  
&
\begin{tabular}{c|c}
28& $\langle 2,1,0,1 \rangle$\\  \hline
29& $\langle 2,1,1,0 \rangle$\\ \hline
30& $\langle 2,2,0,0 \rangle$\\ \hline
31& $\langle 3,0,0,1 \rangle$\\ \hline
32& $\langle 3,0,1,0 \rangle$\\ \hline
33& $\langle 3,1,0,0 \rangle$\\ \hline
34& $\langle 4,0,0,0 \rangle$\\ 
\end{tabular} \\
&&&&
\end{tabular*}
\caption{$4$--dimensional vectors, whose inner sums are $4$, ordered according to the proposed method. }
\label{tab:1}
\end{table}  

The algorithm \textsc{Vector2index} shown in Figure \ref{fig:2} returns the rank of a given vector among the list of
all vectors having the same inner sum. We begin with the the most significant dimension $c_1$, and add index $i$ the
number
of vectors whose first dimension is less than $c_1$. Following the subtraction of the $c_1$ value from the inner sum
$\ell$ and decrementing the dimension by $1$, we continue with the next dimension $c_2$. We repeat the process in the
same way until all dimensions except the last position, which is due to the fact that $c_\sigma$ has no effect on the
index calculation as explained previously. 
\begin{figure}
\begin{center}
\begin{scriptsize}
\begin{tabular}{lr}
\begin{tabular}{|rl|}
\hline
\multicolumn{2}{|l|}{~\textsc{Vector2Index}}\\
\multicolumn{2}{|l|}{~\textsc{Input:} $C=\langle c_1,c_2,\ldots , c_\sigma\rangle$}\\
\multicolumn{2}{|l|}{~\textsc{Output:} Rank of $C$ among all }\\
\multicolumn{2}{|r|}{ vectors with the same inner sum.}\\
~\textsf{1.} & \textsf{$i \leftarrow 0$}\\ 
~\textsf{2.} & \textsf{$\ell = c_1 + c_2 + \ldots + c_\sigma$}\\ 
~\textsf{3.} & \textsf{for $dim \leftarrow 1$  to $\sigma-1$ do}~\\ 
~\textsf{4.} & \quad \textsf{for $j \leftarrow 0$  to $c_{dim}-1$ do}~\\ 
~\textsf{5.} & \quad \quad \textsf{$i \leftarrow i + K(\sigma-dim,\ell-j)$}\\
~\textsf{6.} & \quad \textsf{$\ell \leftarrow \ell - c_{dim}$ }\\
~\textsf{7.} & \textsf{return $i$}\\
&\\
&\\
&\\
&\\
\hline
\end{tabular}
&
\begin{tabular}{|rl|}
\hline
\multicolumn{2}{|l|}{~\textsc{Index2Vector}}\\
\multicolumn{2}{|l|}{~\textsc{Input:} $i$, $\ell$, $\sigma$}\\
\multicolumn{2}{|l|}{~\textsc{Output:} The vector $C=\langle c_1,c_2,\ldots , c_\sigma\rangle$ ranked }\\
\multicolumn{2}{|r|}{  $i^{th}$ in the list of all vectors whose inner sum is $\ell$}\\
~\textsf{1.} & \textsf{for $dim \leftarrow 1$  to $\sigma-1$ do}~\\ 
~\textsf{2.} & \quad \textsf{$c_{dim} \leftarrow 0$}~\\ 
~\textsf{3.} & \quad \textsf{while $(\ell >0)$ do}~\\ 
~\textsf{4.} & \quad \quad \textsf{$t \leftarrow K(\sigma-dim,\ell)$}\\
~\textsf{5.} & \quad \quad \textsf{if $i \geq t $ }\\
~\textsf{6.} & \quad \quad \quad \textsf{$i \leftarrow i - t $ }\\
~\textsf{7.} & \quad \quad \quad \textsf{$c_{dim} \leftarrow c_{dim} + 1 $ }\\
~\textsf{8.} & \quad \quad \quad \textsf{$\ell \leftarrow \ell - 1 $ }\\
~\textsf{9.} & \quad \quad \textsf{else break}\\
~\textsf{10.} & \textsf{$c_\sigma \leftarrow \ell $ }\\
~\textsf{11.} & \textsf{return $C=\langle c_1, c_2, \ldots , c_\sigma \rangle$}\\
\hline
\end{tabular}
\end{tabular}
\caption{The algorithms to find the index of a given vector (\textsc{Vector2index}) and to generate the vector given its
index, inner sum, and dimension (\textsc{index2vector}).}
\label{fig:2}
\end{scriptsize}
\end{center}
\end{figure}

As an example let's follow the algorithm to find the index of $C=\langle 2,1,1,0 \rangle$. We first need compute the
number of vectors whose first dimension is $0$, which is equal to the number of distinct $3$-dimensional vectors whose
inner
sum is $4-0=4$. In the same way we count vectors having $c_1=1$ that corresponds to the number of $3$--dimensional
vectors, where the dimension values sum up to $4-1=3$ now. Those are calculated with to $K(3,4)=15$ and
$K(3,3)=10$ respectively. After being done with $c_1$, we proceed to the next dimension with $c_2=1$. However, before
continuing with second item, we decrement the inner sum from $4$ to $2$ since now first dimension is fixed to $2$.
For the second dimension, we need to consider $2$-dimensional vectors summing up to $2-0=2$, which is $K(2,2)=3$.   
Before visiting the last position we set the inner sum to $4-2-1=1$, and dimension to $1$. Number of one dimensional
vectors having a value less then $c_3=1$ is $K(1,1)=1$. Thus, the index of the input $C$ is $15+10+3+1=29$.

The reverse \textsc{index2vector} algorithm in Figure \ref{fig:2} generates the vector, where the index and the inner
sum as well as the dimension of the vector are given. 
The computation of the values of the vector begins with the most significant position $c_1$. 
We first count all vectors that have a $0$ value at $c_1$, which is computed by $K(\sigma-1,\ell)$. 
If that count is less than or equal to our input index $i$, we decrement $i$ by that count and continue with calculating
the number of vectors that have a value of $1$ at $c_1$, which is equal to $K(\sigma-1,\ell-1)$ now.  We proceed in the
same way until $i$ becomes less than the number of vectors having some value at $c_1$, where that value is actually the
first dimension value we are seeking for. After setting $c_1$, we do the same processing for $c_2$, however, we reduce
the dimension by one and also consider the inner sum to be $\ell - c_1$. All values of the vector $C$ are calculated by
repeating the procedure similarly.

The time complexities of \textsc{vector2index} and  \textsc{index2vector} are $O(c_1 + c_2 + \ldots + c_{\sigma-1})$ 
and  $O(\ell)$ respectively. Note that actually $\ell = c_1 + c_2 + \ldots + c_\sigma$ in \textsc{index2vector}, and
hence, both algorithms are linear with the inner sum of the vector that is being enumerated.

\section{Enumerating $q$-ary sequences}

Let $T=t_1t_2\ldots t_n$ denote a given sequence of length $n$, where the symbols are drawn
from alphabet $\Sigma = \{\epsilon_1, \epsilon_2, \ldots \epsilon_\sigma\}$. 
We propose to factor $T$ into variable length blocks such that symbol $\epsilon_\alpha$ occurs $r$ times. 
Symbol $\epsilon_\alpha$ and $r$ are predetermined fixed parameters of the proposed factorization. 

Careful readers will realize that  if $t_it_{i+1}\ldots t_{i+j}$ is such a factor of $T$,  then $t_{i+j+1}$ is always
equal to $\epsilon_\alpha$. Thus, while coding the sequence we do not need to consider the subsequent characters of the
factors since they are known in advance. 

As an example assume $T=ttgaacgagaagccgtatgaaatgaaaatatcac$ is a given sequence of length $33$ over alphabet
$\Sigma={a,c,g,t}$, where the predetermined parameters of the proposed method are  $\epsilon_\alpha = a$, and
$r=2$. We pad $T$ with $2$ $a$ symbols to make sure that we have a valid block obeying the rule at the end.  
The factorization of $T$ is depicted in Figure \ref{fig:4}. Note that the last two $a$s on block $b_6$ do not belong
to original sequence, which can be determined at decoding phase once the length of the $T$ is known. Hence, the
proposed system requires to explicitly store three parameters as the sequence length $n$, selected symbol
$\epsilon_\alpha$, and the repeat count $r$ of the selected symbol. 
\begin{figure}
\begin{scriptsize}
\begin{center}
\renewcommand{\arraystretch}{1.2}
\begin{tabular}{r@{\hspace{5pt}}cccccccccccccccccccccccccccccccccccc}
&
\multicolumn{7}{c@{\hspace{5pt}}|}{$ttgaacg$} &$a$& 
\multicolumn{8}{|@{\hspace{5pt}}c@{\hspace{5pt}}|}{$gaagccgt$} &$a$&
\multicolumn{4}{|@{\hspace{5pt}}c@{\hspace{5pt}}|}{$tgaa$} &$a$&
\multicolumn{4}{|@{\hspace{5pt}}c@{\hspace{5pt}}|}{$tgaa$} &$a$&
\multicolumn{5}{|@{\hspace{5pt}}c@{\hspace{5pt}}|}{$atatc$} &$a$&
\multicolumn{3}{|@{\hspace{5pt}}c@{\hspace{5pt}}}{$caa$} \\

ID: &
\multicolumn{7}{c@{\hspace{5pt}}}{$b_1$} & & 
\multicolumn{8}{@{\hspace{5pt}}c@{\hspace{5pt}}}{$b_2$} &&
\multicolumn{4}{@{\hspace{5pt}}c@{\hspace{5pt}}}{$b_3$} &&
\multicolumn{4}{@{\hspace{5pt}}c@{\hspace{5pt}}}{$b_4$} &&
\multicolumn{5}{@{\hspace{5pt}}c@{\hspace{5pt}}}{$b_5$} &&
\multicolumn{3}{@{\hspace{5pt}}c@{\hspace{5pt}}}{$b_6$} \\
Length: &
\multicolumn{7}{c@{\hspace{5pt}}}{$7$} & & 
\multicolumn{8}{@{\hspace{5pt}}c@{\hspace{5pt}}}{$8$} &&
\multicolumn{4}{@{\hspace{5pt}}c@{\hspace{5pt}}}{$4$} &&
\multicolumn{4}{@{\hspace{5pt}}c@{\hspace{5pt}}}{$4$} &&
\multicolumn{5}{@{\hspace{5pt}}c@{\hspace{5pt}}}{$5$} &&
\multicolumn{3}{@{\hspace{5pt}}c@{\hspace{5pt}}}{$3$} \\
Sym.Freq.: &
\multicolumn{7}{c@{\hspace{5pt}}}{$\langle 2,1,2,2 \rangle$} & & 
\multicolumn{8}{@{\hspace{5pt}}c@{\hspace{5pt}}}{$\langle 2,2,3,1 \rangle$} &&
\multicolumn{4}{@{\hspace{5pt}}c@{\hspace{5pt}}}{$\langle 2,0,1,1 \rangle$} &&
\multicolumn{4}{@{\hspace{5pt}}c@{\hspace{5pt}}}{$\langle 2,0,1,1 \rangle$} &&
\multicolumn{5}{@{\hspace{5pt}}c@{\hspace{5pt}}}{$\langle 2,1,0,2 \rangle$} &&
\multicolumn{3}{@{\hspace{5pt}}c@{\hspace{5pt}}}{$\langle 2,1,0,0 \rangle$} \\
Perm. ID.: &
\multicolumn{7}{c@{\hspace{5pt}}}{$396$} & & 
\multicolumn{8}{@{\hspace{5pt}}c@{\hspace{5pt}}}{$ 852 $} &&
\multicolumn{4}{@{\hspace{5pt}}c@{\hspace{5pt}}}{$ 11$} &&
\multicolumn{4}{@{\hspace{5pt}}c@{\hspace{5pt}}}{$11 $} &&
\multicolumn{5}{@{\hspace{5pt}}c@{\hspace{5pt}}}{$ 7 $} &&
\multicolumn{3}{@{\hspace{5pt}}c@{\hspace{5pt}}}{$ 2 $} \\

\end{tabular}
\end{center}
\caption{Sample factorization of $T=ttgaacgagaagccgtatgaaatgaaaatatcac$ via the proposed method with parameters
$\epsilon_\alpha = a$, and $r=2$.}
\label{fig:4}
\end{scriptsize}
\end{figure} 

Remembering that $c_\alpha$ indicates occurrence count of $\epsilon_\alpha$ in $T$, the number of blocks with the
proposed factorization is $B = \lfloor \bigl(c_\alpha/(r+1)\bigr) + 1 \rfloor$. Each block can be specified
with three components as the block length, the symbol frequency vector of the block, and the permutation index of the
block. 

We suggest to encode the block lengths via arithmetic coding. When the symbols are ideally identically distributed over
$T$, the frequency of $\epsilon_\alpha$ is nearly $c_\alpha \approx n/\sigma$. We expect to have block lengths mostly in
between $r \cdot c_\alpha$ and $(r+1) \cdot c_\alpha$ symbols, and hence, can be coded/decoded efficiently. 

The frequency vector of each block can be enumerated with the methods described in previous section.
Notice that the $c_\alpha$ value is always $r$ in every block. Thus, it can be excluded from the frequency vector,
which means it is enough to enumerate $\sigma - 1$ values, whose total sum is $r$ less than the block length. For
example, in Figure \ref{fig:4}, the first block is $7$ symbols long and the corresponding symbol frequencies are
$\langle 2,1,2,2\rangle$. Excluding the first dimension that is known to be equal to $2$ in advance, we need to
enumerate $\langle 1,2,2\rangle$. There are $K(3,5)=21$ three dimensional vectors having an inner sum of $5$, which
means we have to spend $\lceil \log 21 \rceil = 5$ bits to indicate the target vector.
\begin{table}
\begin{scriptsize}
\centering
\renewcommand{\arraystretch}{1.2}
\begin{tabular}{c|c||c|c||c|c||c|c}
$PermID$ & Sequence & $PermID$ & Sequence & $PermID$ & Sequence & $PermID$ & Sequence \\ \hline
0 & $aacg$ & 3 & $acga$ &6  & $caag$ & 9  & $gaac$   \\
1 & $aagc$ & 4 & $agac$ &7  & $caga$ & 10 & $gaca$ \\
2 & $acag$ & 5 & $agca$ &8  & $cgaa$ & 11 & $gcaa$  \\ 
\multicolumn{8}{c}{}
\end{tabular} 
\caption{Lexicographically sorted list of all possible sequences over alphabet $\Sigma = \{a,c,g,t\}$ with a
sample frequency vector of  $C=\langle 2,1,1,0 \rangle$ .}
\label{table5} 
\end{scriptsize}
\end{table}

When we know the frequencies of each symbol, a block may be indicated by its rank among the lexicographically ordered
list of all its permutations. The algorithms to find the rank of a given sequence
(\textsc{Sequence2PermIndex}), and its reverse operation (\textsc{PermIndex2Sequence}) as generating the sequence given
the symbol frequencies and the rank are given in Figure \ref{fig:3}. The permutation
indices of the blocks shown in figure \ref{fig:3} are computed according to these procedures.

In \textsc{Sequence2PermIndex}, we initially set index to $0$ and scan the input sequence from left to right. 
At each position, we count how many possible strings precede input $S$  and add this value to index. As an example, let
us find the rank of $agca$. Possible permutations can be viewed in Table \ref{table5}. First character $a$ is the least
significant symbol of the alphabet and hence no symbols precede it. Next we investigate the number of strings with
prefixes $aa$ and $ac$, which makes index $2+2=4$. The last step is to count strings with prefix $aga$, which is $1$.
Thus, we compute that $agca$ has rank $5$ in a zero-based list. 

The \textsc{PermIndex2Sequence} is the reverse of this process. The time complexities of both algorithms are $O(\ell)$,
where $\ell$ is the length of input/output string.
\begin{figure}
\begin{center}
\begin{scriptsize}
\begin{tabular}{lr}
\begin{tabular}{|rl|}
\hline
\multicolumn{2}{|l|}{~\textsc{Sequence2PermIndex}}\\
\multicolumn{2}{|l|}{~\textsc{Input:} Sequence $S=s_1s_2\ldots s_\ell$}\\
\multicolumn{2}{|l|}{~\textsc{Output:} The rank $pid$ of $S$ in the }\\
\multicolumn{2}{|r|}{ordered list of all sequences }\\
\multicolumn{2}{|r|}{with the same symbol frequencies}\\
~\textsf{1.} & \textsf{for $j \leftarrow 1$  to $\sigma$ do $c_j \leftarrow 0$}~\\
~\textsf{2.} & \textsf{for $j \leftarrow 1$  to $\ell$ do }~\\  
~\textsf{3.} & \quad \textsf{$k \leftarrow symbolID(s_j)$}~\\ 
~\textsf{4.} & \quad \textsf{$c_k \leftarrow c_k + 1$}~\\
~\textsf{5.} & \textsf{for $a \leftarrow 1$  to $\ell$ do}~\\ 
~\textsf{6.} & \quad \textsf{$k \leftarrow symbolID(s_a)$}~\\ 
~\textsf{7.} & \quad \textsf{for $j \leftarrow 1$ to $k-1$}~\\ 
~\textsf{8.} & \quad \quad \textsf{if $(c_j >0)$}~\\ 
~\textsf{9.} & \quad \quad \quad \textsf{ $c_j \leftarrow c_j - 1$}\\ 
~\textsf{10.} & \quad \quad \quad \textsf{$t \leftarrow \frac{\ell!}{c_1!\cdot c_2! \cdot \ldots \cdot c_\sigma!}$}\\
~\textsf{11.} & \quad \quad \quad \textsf{$pid \leftarrow pid + t$}\\
~\textsf{12.} & \quad \quad \quad \textsf{ $c_j \leftarrow c_j + 1$}\\ 
~\textsf{13.} & \quad \textsf{ $c_k \leftarrow c_k - 1$}~\\ 
~\textsf{14.} & \textsf{return $pid$}\\
&\\
\hline
\end{tabular}
&
\begin{tabular}{|rl|}
\hline
\multicolumn{2}{|l|}{~\textsc{PermIndex2Sequence}}\\
\multicolumn{2}{|l|}{~\textsc{Input:} $pid$, $C=\langle c_1,c_2,\ldots , c_\sigma\rangle$}\\
\multicolumn{2}{|l|}{~\textsc{Output:} The $pid^{th}$ sequence $S$ in the }\\
\multicolumn{2}{|r|}{ordered list of all sequences }\\
\multicolumn{2}{|r|}{with the same symbol frequencies}\\
~\textsf{1.} & \textsf{$temp \leftarrow 0$}\\ 
~\textsf{2.} & \textsf{$\ell = c_1 + c_2 + \ldots + c_\sigma$}\\ 
~\textsf{3.} & \textsf{for $a \leftarrow 1$  to $\ell$ do}~\\ 
~\textsf{4.} & \quad \textsf{for $j \leftarrow 1$  to $\sigma$ do}~\\ 
~\textsf{5.} & \quad \quad \textsf{if ($c_j > 0$)}\\
~\textsf{6.} & \quad \quad \quad \textsf{$c_j \leftarrow c_j -1$}\\
~\textsf{7.} & \quad \quad \quad \textsf{$t \leftarrow \frac{\ell!}{c_1!\cdot c_2! \cdot \ldots \cdot c_\sigma!}$}\\
~\textsf{8.} & \quad \quad \quad \textsf{$temp \leftarrow temp + t$}\\
~\textsf{9.} & \quad \quad \quad \textsf{if ($pid \geq temp$)}\\
~\textsf{10.} & \quad \quad \quad \quad \textsf{$c_j \leftarrow c_j +1$}\\
~\textsf{11.} & \quad \quad \quad \textsf{else}\\
~\textsf{12.} & \quad \quad \quad \quad \textsf{$s_a \leftarrow j$}\\
~\textsf{13.} & \quad \quad \quad \quad \textsf{$temp \leftarrow temp - t$}\\
~\textsf{14.} & \quad \quad \quad \quad \textsf{$pid \leftarrow pid - temp$}\\
~\textsf{15.} & \textsf{return $s_1s_2\ldots s\ell$}\\
\hline
\end{tabular}
\end{tabular}
\caption{The algorithms to find the permutation index of a given sequence (\textsc{Sequence2Permindex}) and to generate
the sequence given its permutation index along with symbol frequencies  (\textsc{permindex2sequence}).}
\label{fig:3}
\end{scriptsize}
\end{center}
\end{figure} 

Notice that the permutation index of a block need not to be coded when the corresponding frequency index includes only
one non-zero dimension, which means all symbols are same within that block, and that symbol is the one that is non-zero
in the $C$ vector.

\textbf{Analysis:} The enumerative encoding of block $b_i$ of length $|b_i|$, $0<i<B$, requires $\log |b_i|$ bits for
length, $\log K(\sigma-1,|b_i|)$ bits for enumeration of frequency vector, and $H^f_0(b_i)$ bits for the permutation
index, where $H^f_0(b_i)$ represents the zero-order finite set entropy of the sequence. Hence total space requirement
is
\begin{equation}
 \sum_{i=1}^{B} \log |b_i| + \log K(\sigma-1,|b_i|) + H^f_0(b_i)
\label{eq:7}
\end{equation}

It may be argued to use fixed--length blocks to remove the overhead of block length encoding. In this case, we
partition the sequence into $B^\prime$ factors of length $\frac{n}{B^\prime}$, and the total bits required becomes 
\begin{equation}
 \sum_{i=1}^{B^\prime} \log K(\sigma,\frac{n}{B^\prime}) + H^f_0(b_i)
\end{equation}

When we assume the number of blocks are set to be nearly equal on fixed and variable length models, although
we remove the $\log |b_i|$ term in fixed--length model, the frequency vector we need to encode in this case is $\sigma$
dimensional, where it is $\sigma-1$ dimensional in the variable model. This trade--off is investigated experimentally,
and it is observed  that variable length model performs better than fixed length on  DNA sequences (see next section).

Another point is the selection of the parameters $\epsilon_\alpha$ and $r$ for the proposed variable--length model. If
one chooses a frequent symbol, the lengths of the blocks become shorter and the number of blocks get larger. Similarly
small $r$ values increase block count and vice versa. The trade--off between the block length and block number is also
experimentally analyzed on DNA sequences, where it is observed that selecting rare symbols give better results.

\section{Implementation and experimental analysis on DNA sequences}

We have implemented an enumerative coding library with C++ that can be freely downloaded from
\url{www.busillis.com/EnumCodeLib.zip}. Thanks to GNU multiple precision arithmetic library\footnote{\url{gmplib.org}}
that lets us to overcome the big number problems due to factorial/combination/permutation calculations. Using this
library we have also implemented the proposed variable--length factorization as well as the fixed--length scheme, and
computed the number of bits required with each method on a series of DNA sequence files that are mostly used as the
benchmark corpus on DNA compression algorithms. The results are given in table \ref{tableRes}. 
\begin{table}
\begin{center}
\begin{scriptsize}
\renewcommand{\arraystretch}{1.2}
\begin{tabular}{rccccc|c|c|ccc}
 & & & & & & Finite-set  & Fixed-Len. & Variable-Len. &  & Avg. Block\\
File name & Size & $a$ & $c$ & $g$ & $t$ & $H_0$ Entropy & Entropy & Entropy & $\epsilon_\alpha$ & Length\\
chmpxx	&	121024	&	42896	&	17309	&	17556	&	43263	&	1.866	&	1.871
&	1.845	&	g & 504\\
chntxx	&	155844	&	47824	&	29991	&	28992	&	49037	&	1.957	&	1.969
&	1.953	&	g & 402\\
hehcmv	&	229354	&	49475	&	64911	&	66192	&	48776	&	1.985	&	1.975
&	1.977	&	a & 796\\
humdyst	&	38770	&	12001	&	7161	&	7011	&	12597	&	1.946	&	1.956
&	1.950	&	c & 692\\
humghcs	&	66495	&	17311	&	16271	&	16441	&	16472	&	1.999	&	2.028
&	1.999	&	a & 493\\
humhbb	&	73308	&	22068	&	14146	&	14785	&	22309	&	1.967	&	1.972
&	1.956	&	t & 424\\
humhdab	&	58864	&	13422	&	14846	&	15906	&	14690	&	1.997	&	1.999
&	1.977	&	t & 516\\
humprtb	&	56737	&	15689	&	11281	&	11599	&	18168	&	1.970	&	1.990
&	1.963	&	g & 630\\
mpomtcg	&	186609	&	53206	&	39215	&	39924	&	54264	&	1.983	&	1.998
&	1.983	&	c & 614\\
mtpacga	&	100314	&	35804	&	13428	&	16724	&	34358	&	1.879	&	1.882
&	1.883	&	c & 955\\
vaccg	&	191737	&	63921	&	32010	&	32030	&	63776	&	1.919	&	1.929
&	1.926	&	c & 770 \\ \hline \hline
\multicolumn{6}{r|}{\textbf{Average entropy in bits/base}} & \textbf{1.952} & \textbf{1.961} & \textbf{1.946} & & \\
\multicolumn{6}{c}{} \\
\end{tabular}
\end{scriptsize}
\end{center}
\caption{Comparison of results obtained by fixed-- and variable--length enumerative coding on sample DNA sequences. The
block length for the fixed--length scheme is 2048, where the $r$ parameter is $128$ on variable--length }
\label{tableRes}
\end{table} 

We tested windows of lengths $4$, $8, 16,32, \ldots, 2048$ for the fixed--length scheme. We observed that best results
have been observed by the longest size $2048$. 
We followed a two--pass procedure for the variable--length enumeration. In the first pass we calculated the block
lengths, and in the second phase we computed the frequency-vector enumerations as well as the permutation identities of
each block. The \textit{bits per base} entropy values has been computed with 
$\frac{1}{filesize} \sum_{i=1}^{B} \log |b_i| + \lceil \log K(\sigma-1,|b_i|) \rceil + \lceil \frac{|b_i|!}{c_1!\cdot
c_2! |ldots c_q!}\rceil$. We used $r$ values of $4,8,16,32,64$, and $128$ for each possible
$\epsilon_\alpha=\{a,c,g,t\}$.  We observed best
results with $r=128$ and when $\epsilon_\alpha$ is usually one of the least frequent symbol. The corresponding average
block lengths are depicted in the table. 

Proposed variable--length model achieved better results than the naive fixed--length model. Note that the entropy
measured with the variable--length scheme is on the average better than the possible theoretical zero--order
finite set entropy of the target sequence. That is due to saving one symbol per block with the introduced model, and
the skewed distribution of symbols in blocks, which enabled us to identify them with less bits via permutation indices.
Another important point is proposed variable-length block enumeration scheme achieved  better results with shorter 
blocks than the naive fixed--length scheme, which is also important as it requires more time to compute the vector and
permutation indices on longer lengths. 

\section{Conclusions}

We have investigated the enumerative encoding of $q$--ary sequences by 
introducing a new method, where we factorize the input sequence into variable size blocks and then
represent each block via its length, enumerated coding of its symbol frequency vector, and the regarding permutation
index. We have described the necessary procedures required for vector-to-index and sequence-to-index conversions along
with their reverse operations for appropriate decoding.  
We have devised experiments on DNA sequences to compare the performances of the proposed variable-- versus naive
fixed--length enumerations with respect to the zero order finite set entropies of the sample sequences. 

The enumeration library used in this
study, which we have implemented via the GNU multiple precision arithmetic toolbox,  is publicly
available\footnote{\textit{EnumCodeLib} at \url{www.busillis.com/EnumCodeLib.zip}}. We believe that one of the main
reasons for relatively less attention paid on enumerative coding is because of the lack of such publicly available
resources. Notice that one can find many libraries for arithmetic, Huffman, or dictionary based Liv-Zempel
factorizations, however, not for enumerative coding. 

It would be interesting to investigate more options on enumeration based approaches in data compression, as well as
text indexing. Benefiting from enumerative coding in information theoretic analyses of biological sequences is
yet another direction.

\bibliography{tapas12}

\end{document}